\documentclass[conference]{IEEEtran}
\IEEEoverridecommandlockouts
\usepackage{cite}
\usepackage{amsmath,amssymb,amsfonts}
\usepackage{algorithmic}
\usepackage{graphicx}
\usepackage{textcomp}
\usepackage{xcolor}
\def\BibTeX{{\rm B\kern-.05em{\sc i\kern-.025em b}\kern-.08em
    T\kern-.1667em\lower.7ex\hbox{E}\kern-.125emX}}

\usepackage{color}
\usepackage{amsthm}
\usepackage{arydshln}
\newtheorem{theorem}{Theorem}
\newtheorem{corollary}{Corollary}[theorem]
\newtheorem{lemma}[theorem]{Lemma}
\theoremstyle{definition}
\newtheorem{definition}{Definition}[section]
\newtheorem*{remark}{Remark}
\newtheorem{proposition}{Proposition}[section]
\usepackage{physics}
\usepackage{bm}
\usepackage{tikz}
\ifCLASSOPTIONcompsoc
\usepackage[caption=false,font=normalsize,labelfon
t=sf,textfont=sf,subrefformat=parens,labelformat=parens]{subfig}
\else
\usepackage[caption=false,font=footnotesize,subrefformat=parens,labelformat=parens]{subfi
g}
\fi

\usepackage{float}
\usepackage{pgfplots}
\usetikzlibrary{arrows.meta}

\newcommand{\mcl}[1]{\mathcal{#1}}

\newcommand{\bd}{\begin{definition}}
\newcommand{\ed}{\end{definition}}
\newcommand{\bl}{\begin{lemma}}
\newcommand{\el}{\end{lemma}}
\newcommand{\bt}{\begin{theorem}}
\newcommand{\et}{\end{theorem}}
\newcommand{\bc}{\begin{corollary}}
\newcommand{\ec}{\end{corollary}}
\newcommand{\bp}{\begin{proposition}}
\newcommand{\ep}{\end{proposition}}
\newcommand{\ba}{\begin{axiom}}
\newcommand{\ea}{\end{axiom}}
\newcommand{\be}{\begin{example}}
\newcommand{\ee}{\end{example}}
\newcommand{\br}{\begin{remark}}
\newcommand{\er}{\end{remark}}
\newcommand{\bx}{\begin{exercise}}
\newcommand{\ex}{\end{exercise}}

\begin{document}

\title{A Modified MWPM Decoding Algorithm for Quantum Surface Codes Over Depolarizing Channels
}

\author{
\IEEEauthorblockN{Yaping Yuan}
\IEEEauthorblockA{\textit{Department of Electrical Engineering} \\
\textit{National Tsing Hua University}\\
Hsinchu 30013, Taiwan \\
yapingyuan@mx.nthu.edu.tw}
\and
\IEEEauthorblockN{Chung-Chin Lu}
\IEEEauthorblockA{\textit{Department of Electrical Engineering} \\
\textit{National Tsing Hua University}\\
Hsinchu 30013, Taiwan \\
cclu@ee.nthu.edu.tw}
}

\maketitle

\begin{abstract}
Quantum Surface codes are a kind of quantum topological stabilizer codes whose stabilizers and qubits are geometrically related. Due to their special structures, surface codes have great potential to lead people to large-scale quantum computation. In the minimum weight perfect matching (MWPM) decoding of surface codes, the bit-flip errors and phase-flip errors are assumed to be independent for simplicity. However, these two kinds of errors are likely to be correlated in the real world. In this paper, we propose a modification to MWPM decoding for surface codes to deal with the noise in depolarizing channels where bit-flip errors and phase-flip errors are correlated. With this modification, we obtain thresholds of 17\% and 15.3\% for the surface codes with mixed boundaries and the surface codes with a hole, respectively.
\end{abstract}

\section{Introduction}
Quantum error-correcting codes play a very important role in the development of quantum computation since the inherent sensitivity of quantum systems to noise. Stabilizer codes are a class of quantum error-correcting codes that have a strong connection with classical error-correcting codes. The code space of a stabilizer code is determined by the so-called  stabilizers. Topological codes are a class of stabilizer codes whose stabilizers and data qubits are topologically related. It is believed that topological codes have great potential to be implemented on large scales due to their special structures. Therefore, topological codes have gained a lot of attention in recent years. The surface codes are a family of topological codes defined on a 2D lattice of qubits \cite{boundary,hole}.

Various decoders for surface codes have been developed in recent years, such as the decoders based on belief-propagation (BP) \cite{BP1,BP2}, union-find (UF) \cite{UF1}, and matrix product states (MPS) \cite{MPS1,TN}. The most standard decoder for surface codes is the minimum weight perfect matching (MWPM) decoder. When the bit-flip errors and the phase-flip errors are assumed to be uncorrelated, the quantum maximum likelihood decoding (QMLD) of surface codes can be reduced to problems of finding an minimum weight perfect matching on a graph. However, the depolarizing noise model, where the bit-flip errors and the phase-flip errors are correlated, is closer to the real world. In this paper, we propose a modification to the  vanilla MWPM decoding of surface codes to deal with the noise in depolarizing channels.

Our method is based on iteratively reweighting the dual lattice and the primal lattice with the correction pattern on the other lattice. 
Similar methods were proposed in \cite{fowler2013optimal,OneRoundReweighting}, but it was not shown whether it's possible that the weight of the correction operation will increase along the iterations. In this paper, besides showing how the iteratively reweighted MWPM decoding works, we will also prove that the weight of the correction operation will never increase along the iterations. 

This paper is organized as follows. In Section II, we review the structure of surface codes and the MWPM decoder. In Section III, we discuss our modification to the MWPM decoder. In Section IV, we provide the simulation results for the IRMWPM decoding. Section V concludes this paper.

\section{Basics of Surface Codes}
\subsection{Structure of Surface Codes}
In this paper, we describe a surface code in a similar way as \cite{dual lattice} does. A surface code is defined on a square lattice and every edge on this lattice is associated with a qubit. There are two types of stabilizer generators: plaquette stabilizer generators and vertex stabilizer generators. Every plaquette is associated with a plaquette stabilizer generator. A plaquette stabilizer generator consists of a tensor product of Pauli $Z$ operators acting on the qubits that lie on the plaqeutte's boundary, as illustrated in Fig. \subref*{Stabilizer:sub1}. For every vertex, there is a vertex stabilizer generator which consists of a tensor product of Pauli $X$ operators acting on the qubits adjacent to the vertex, as shown in Fig. \subref*{Stabilizer:sub2}.

\begin{figure}[h]
\centering
\subfloat[]{
  \begin{tikzpicture}[scale=0.6, every node/.style={scale=0.7}]
	\filldraw[fill=blue!30, draw= blue!30] (0,1)rectangle(1,2);
	\draw[step=1](-0.5,-0.5) grid(3.5,3.5);
	\filldraw[fill=blue, draw= blue] (0,1.5) circle(0.2);
	\node [white] at (0,1.5) {\textbf{Z}};
	\filldraw[fill=blue, draw= blue] (1,1.5) circle(0.2);
	\node [white] at (1,1.5) {\textbf{Z}};
	\filldraw[fill=blue, draw= blue] (0.5,1) circle(0.2);
	\node [white] at (0.5,1) {\textbf{Z}};
	\filldraw[fill=blue, draw= blue] (0.5,2) circle(0.2);
	\node [white] at (0.5,2) {\textbf{Z}};
  \end{tikzpicture}
  \label{Stabilizer:sub1}}
  \qquad
\subfloat[]{
  \begin{tikzpicture}[scale=0.6, every node/.style={scale=0.7}]
	\draw[step=1](-0.5,-0.5) grid(3.5,3.5);
	\draw[red, line width = 1mm] (1.5,1) -- (2.5,1);
	\draw[red, line width = 1mm] (2,0.5) -- (2,1.5);
	
	\filldraw[fill=red, draw= red] (1.5,1) circle(0.2);
	\node [white] at (1.5,1) {\textbf{X}};
	\filldraw[fill=red, draw= red] (2.5,1) circle(0.2);
	\node [white] at (2.5,1) {\textbf{X}};
	\filldraw[fill=red, draw= red] (2,1.5) circle(0.2);
	\node [white] at (2,1.5) {\textbf{X}};
	\filldraw[fill=red, draw= red] (2,0.5) circle(0.2);
	\node [white] at (2,0.5) {\textbf{X}};
  \end{tikzpicture}
  \label{Stabilizer:sub2}}
\caption{(a) A plaquette stabilizer generator. (b) A vertex stabilizer generator.}
\label{Stabilizer}
\end{figure}
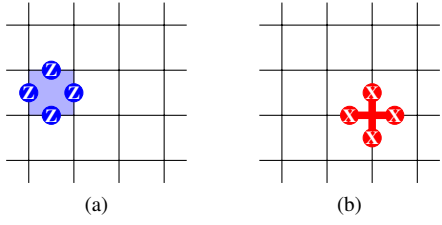

There are two main types of surface codes, one is built on a lattice with mixed boundaries \cite{boundary}, and the other is built on a lattice with holes (or defects) \cite{hole}. Surface codes with mixed boundaries are constructed on a lattice surrounded by two pairs of different boundaries, as shown in Fig. \subref*{Mixed Boundary}. Surface codes with a hole are constructed on a lattice where a plaquette stabilizer generator in the middle of the lattice is removed, as shown in Fig. \subref*{A single-cut qubit}. Note that the size of a hole is not necessarily $1\times 1$. In Fig. \ref{Surface code}, the original lattices are called the primal lattices, and the lattices depicted in dashed lines are called the dual lattices \cite{dual lattice}.

\begin{figure}[h]
\centering
\subfloat[]{
\begin{tikzpicture}[scale=0.6]
\draw[step=1,dash pattern=on 2pt off 3pt, yshift=0.5cm](0,-0.5) 				grid(3,2.5);
\draw[step=1,xshift=0.5cm](-0.5,0) grid(2.5,3);
\end{tikzpicture}
\label{Mixed Boundary}
}\qquad
\subfloat[]{
\begin{tikzpicture}[scale=0.6]
\draw[step=1,dash pattern=on 2pt off 2pt,xshift=0.5cm, yshift=0.5cm](-0.5,-0.5) 				grid(2.5,2.5);
\filldraw[fill=black!50, draw= black!50] (1,1)rectangle(2,2);
\draw[step=1](0,0) grid(3,3);
\end{tikzpicture}
\label{A single-cut qubit}
}
\caption{(a) A surface code with mixed boundaries. (b) A surface code with a hole.}
\label{Surface code}
\end{figure}
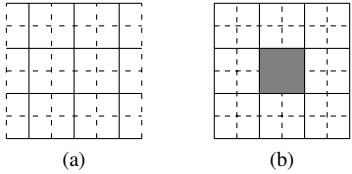

\subsection{Syndromes of Surface Codes}

For a stabilizer code, each stabilizer generator corresponds to an element of the syndrome vector. For an error $E$, the stabilizer generators that anti-commute with $E$ will give a $1$ in the syndrome vector, otherwise $0$. For simplicity, we call the stabilizer generators that give nonzero syndrome elements in a surface codes ``syndrome nodes". 

Suppose $E_Z$ is a tensor product of Pauli $Z$ errors. Since a Pauli $X$ anti-commutes with a Pauli $Z$, if we express $E_Z$ as strings on the primal lattice, then the syndrome nodes corresponding to $E_Z$ are the endpoints of those strings, as shown in Fig. \subref*{z string}. Similarly, we can express $X$-type errors as strings on the dual lattice, and the corresponding syndrome nodes are the endpoints of those strings, as shown in Fig. \subref*{string on dual lattice}. Since a Pauli $Y$ anti-commutes with both a Pauli $X$ and a Pauli $Z$, we can treat a $Y$ error as a combination of an $X$ error and a $Z$ error.

\begin{figure}[h]
\centering
\subfloat[]{
\begin{tikzpicture}[scale=0.6]
\draw[step=1](0,0) grid(3,3);
\filldraw[fill=black!50, draw= black] (1,1)rectangle(2,2);

\draw [blue, line width = 1mm] (0,2) -- (1,2);
\filldraw[fill=blue, draw= blue] (0,2)circle(0.15);
\filldraw[fill=blue, draw= blue] (1,2)circle(0.15);

\draw [blue, line width = 1mm] (2,2) -- (3,2);
\draw [blue, line width = 1mm] (3,2) -- (3,0);
\filldraw[fill=blue, draw= blue] (2,2)circle(0.15);
\filldraw[fill=blue, draw= blue] (3,0)circle(0.15);
\end{tikzpicture}
\label{z string}
}
\qquad
\subfloat[]{
\begin{tikzpicture}[scale=0.6, every node/.style={scale=0.9}]
\draw[step=1](0,0) grid(3,3);
\draw[step=1,dash pattern=on 2pt off 2pt,xshift=0.5cm, yshift=0.5cm](-0.5,-0.5) grid(2.5,2.5);
\filldraw[fill=black!50, draw= black] (1,1)rectangle(2,2);

\draw [red, line width = 1mm] (0.5,0.5) -- (0.5,2.5);
\filldraw[fill=red, draw= red] (0.5,0.5)circle(0.15);
\filldraw[fill=red, draw= red] (0.5,2.5)circle(0.15);

\draw [red, line width = 1mm] (2.5,0) -- (2.5,1.5);
\filldraw[fill=red, draw= red] (2.5,1.5)circle(0.15);

\draw [red, line width = 1mm] (1.5,2) -- (1.5,2.5);
\draw [red, line width = 1mm] (1.5,2.5) -- (2.5,2.5);
\filldraw[fill=red, draw= red] (2.5,2.5)circle(0.15);

\node [red] at (0.25,1.75) { $s_1$};
\node [red] at (1.75,2.75) { $s_2$};
\node [red] at (2.8,0.75) { $s_3$};
\end{tikzpicture}
\label{string on dual lattice}
}
\caption{(a) A tensor product of $Z$ errors is depicted in blue lines and the corresponding syndrome nodes are depicted in blue circles. (b) A tensor product of $X$ errors is depicted in red lines and the corresponding syndrome nodes are depicted in red circles. The string $s_1$ has two plaquette generators that anti-commute with it, but $s_2$ and $s_3$ both have only one plaquette operator that anti-commutes with them.}
\label{error_strings}
\end{figure}
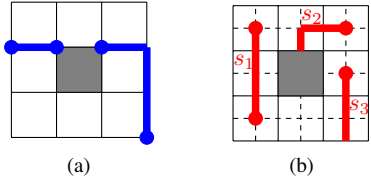

For a surface code with mixed boundaries, its code distance is the distance between the two sides. Therefore, the code distance of Fig. \subref*{Mixed Boundary} is $4$. For a surface code with a hole, operators that commute with all stabilizers but not stabilizers themselves are either loops of $Z$ operators that wind around the hole or strings of $X$ operators that connect the inner and outer boundaries. Let the number of qubits on the shortest path between the inner boundary and the outer boundary be $d_b$ and the number of qubits around the hole be $d_h$, then the code distance of a surface code with a hole is $d=\min(d_b,d_h)$. Therefore, the code distance of Fig. \subref*{A single-cut qubit} is $2$.

\subsection{MWPM decoding of Surface Codes}

Since the syndrome of a surface code can be viewed as nodes on the primal lattice and the dual lattice, the quantum maximum likelihood decoding can be reduced to the problem that finds the most likely string patterns with the same syndrome nodes. How to choose the most likely correction strings depends on the noise models. 

Suppose that $X$ errors and $Z$ errors are independent and a $Y$ error is considered as a combination of an $X$ error and a $Z$ error, then we can decode $X$ errors and $Z$ errors separately. To decode $Z$ errors only, we just need to find the strings on the primal lattice with the minimum weight such that connect all the syndrome nodes on the primal lattice. It's similar for the decoding of $X$ type errors, but the lattice we are working on is the dual lattice instead. Therefore, the decoding of a surface code can be regarded as two minimum weight perfect matching problems. Although the number of syndrome nodes may be odd, with some modifications, the decoding can still be reduced to MWPM problems. The noise model where $X$ errors and $Z$ errors are independent to each other is called the uncorrelated noise model. To solve an MWPM problem, we can use a well known algorithm developed by Jack Edmonds and known as the blossom algorithm \cite{blossom}. The time complexity of the blossom algorithm for a graph $G=(V,E)$ is $O(|V|^3)$. Therefore, the time complexity of the MWPM decoder is $O(n^3)$.

\section{Iteratively reweighted MWPM Decoding of Surface Codes}
The depolarizing noise model is the most considered noise model in quantum error-correction. In a depolarizing channel, each qubit has the probability of $1-\epsilon$ to remain untainted and the probability of $\frac{\epsilon}{3}$ to be affected by $X$, $Y$, and $Z$, respectively. Therefore, if we view a $Y$ error as a combination of $X$ and $Z$, the conditional probability $P(X|Z)=0.5$, so $X$ errors and  $Z$ errors are not independent to each other.
  
As shown in Fig. \ref{rewight example},  if we use the MWPM decoding, we will get an decoding result as Fig. \subref*{MWPM on DN}. However, if the noise model we are considering is the depolarizing noise model, the decoding result of QMLD should be Fig. \subref*{Real QMLD on DN}.

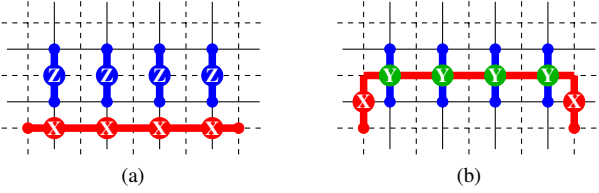
\begin{figure}[h]
     \centering
     \subfloat[]{
		\begin{tikzpicture}[scale=0.7, every node/.style={scale=0.75}]
		\draw[step=1](0.1,0.1) grid(4.9,2.9);
		%\draw[dotted] (0,0) -- (5,0);
		%\draw[dotted] (0,3) -- (5,3);
		%\draw[dotted] (5,0) -- (5,3);
		%\draw[dotted] (0,0) -- (0,3);
		\draw[step=1,dash pattern=on 2pt off 2pt,xshift=0.5cm, yshift=0.5cm](-0.5,-0.5) grid(4.5, 2.5);

		\filldraw [fill=red,draw = red](0.5,0.5) circle(0.1);
		\filldraw [fill=red,draw = red](4.5,0.5) circle(0.1);
		\foreach \x in {1,2,3,4}
			\foreach \y in {1,2}
				\filldraw [fill=blue,draw = blue](\x,\y) circle(0.1);
		\draw [red,line width=1mm] (0.5,0.5) -- (4.5,0.5);
		\foreach \x in {1,2,3,4}
			\draw [blue,line width=1mm] (\x,1) -- (\x,2);
		\foreach \x in {1,2,3,4}{
			\filldraw [fill=red,draw = red](\x,0.5) circle(0.2);
			\node [white] at (\x,0.5) {\textbf{X}};
		}
		\foreach \x in {1,2,3,4}{
			\filldraw [fill=blue,draw = blue](\x,1.5) circle(0.2);
			\node [white] at (\x,1.5) {\textbf{Z}};
		}
	\end{tikzpicture}
		\label{MWPM on DN}}
\qquad
    \subfloat[]{
		\begin{tikzpicture}[scale=0.7, every node/.style={scale=0.75}]
		\draw[step=1](0.1,0.1) grid(4.9,2.9);
		%\draw[dotted] (0,0) -- (5,0);
		%\draw[dotted] (0,3) -- (5,3);
		%\draw[dotted] (5,0) -- (5,3);
		%\draw[dotted] (0,0) -- (0,3);
		\draw[step=1,dash pattern=on 2pt off 2pt,xshift=0.5cm, yshift=0.5cm](-0.5,-0.5) grid(4.5, 2.5);

		\filldraw [fill=red,draw = red](0.5,0.5) circle(0.1);
		\filldraw [fill=red,draw = red](4.5,0.5) circle(0.1);
		\foreach \x in {1,2,3,4}
			\foreach \y in {1,2}
				\filldraw [fill=blue,draw = blue](\x,\y) circle(0.1);
	
		\draw [red,line width=1mm] (0.5,0.5) -- (0.5,1.5);
		\draw [red,line width=1mm] (0.5,1.5) -- (4.5,1.5);
		\draw [red,line width=1mm] (4.5,1.5) -- (4.5,0.5);
		\foreach \x in {1,2,3,4}
			\draw [blue,line width=1mm] (\x,1) -- (\x,2);
		
		\foreach \x in {1,2,3,4}{
			\filldraw [fill=green!70!black,draw = green!70!black](\x,1.5) circle(0.2);
			\node [white] at (\x,1.5) {\textbf{Y}};
		}
		\filldraw [fill=red,draw = red](0.5,1) circle(0.2);
		\node [white] at (0.5,1) {\textbf{X}};
		\filldraw [fill=red,draw = red](4.5,1) circle(0.2);
		\node [white] at (4.5,1) {\textbf{X}};
		\end{tikzpicture}		
		\label{Real QMLD on DN}}
    \caption{Let the small circles be the syndrome nodes. (a) is a QMLD over the uncorrelated noise model and (b) is a QMLD over the depolarizing noise model.}
    \label{rewight example}
\end{figure}

In Fig. \subref*{MWPM on DN}, we have $4$ $X$ operators and $4$ $Z$ operators, so the total weight of this correction is $8$. In Fig. \subref*{Real QMLD on DN}, although there are $6$ $X$ operators and $4$ $Z$ operators, we have $4$ $Y$ and $2$ $X$ under the view point of the depolarizing noise model. Since the total weight in Fig. \subref*{Real QMLD on DN} is only $6$, it is better than Fig. \subref*{MWPM on DN} when the noise model is the depolarizing noise model. 

Suppose that the correction strings on the primal lattice is fixed, we can find that if a string on the dual lattice touches a string on the primal lattice, the intersection does not increase the total weight, since it just turns a single $Z$ correction into a single $Y$ correction. Therefore, the shortest path from one syndrome node on the dual lattice to another is not necessarily the string that can minimize the total weight.  

However, if we reweight the edges on the dual lattice that touches the correction strings on the primal lattice to $0$, the shortest path between the two syndrome nodes on the reweighted dual lattice is the correction string that causes the least extra total weight. As shown in Fig. \ref{the shortest path}, the shortest path between the two syndrome nodes on the reweighted dual lattice is now $P_2$ instead of $P_1$. Therefore, when the correction of $Z$-type error is fixed, finding the MWPM on the reweighted dual lattice can give us the error pattern that minimizes the total weight.

\begin{figure}[h]
	\centering
	\begin{tikzpicture}[scale=0.7, every node/.style={scale=0.8}]
	\draw[black!20, step=1](0.1,0.1) grid(4.9,2.9);
	%\draw[dotted] (0,0) -- (5,0);
	%\draw[dotted] (5,0) -- (5,3);
	\draw[step=1,dash pattern=on 2pt off 2pt,xshift=0.5cm, yshift=0.5cm](-0.5,-0.5) grid(4.5, 2.5);

	\filldraw [fill=red,draw = red](0.5,0.5) circle(0.1);
	\filldraw [fill=red,draw = red](4.5,0.5) circle(0.1);
	\foreach \x in {1,2,3,4}
		\foreach \y in {1,2}
			\filldraw [fill=blue!25,draw = blue!15](\x,\y) circle(0.1);
	\draw [red,line width=1mm] (0.7,0.5) -- (4.3,0.5);
	\foreach \x in {1,2,3,4}{
		\draw [blue!25,line width=1mm] (\x,1) -- (\x,2);
		\node at (\x,1.75) {$0$};
		\node at (\x,0.75) {$1$};
	}
	\node at (0.2,1) {$1$};
	\node at (4.8,1) {$1$};
	\node [red] at (2,0) {$P_1$};
	\draw [red,line width=1mm] (0.5,0.7) -- (0.5,1.5);
	\draw [red,line width=1mm] (0.5,1.5) -- (4.5,1.5);
	\draw [red,line width=1mm] (4.5,1.5) -- (4.5,0.7);
	\node [red] at (4.8,1.8) {$P_2$};

	\end{tikzpicture}
	\caption{The reweighted dual lattice.}
	\label{the shortest path}
\end{figure}
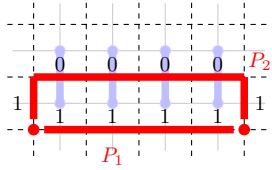

Since the $Z$ correction in a QMLD may not be an MWPM on the primal lattice, we can repeat this process more than one times to give us a better decoding result. We can use an MPWM on the reweighted dual lattice to reweight the primal lattice with the similar way, and then use the new MWPM on the reweighted primal lattice to reweight the original dual lattice again. In this paper, we will prove that no matter how many times we repeat this reweighting process, the total weight will only become smaller and smaller or remain the same.

Let $\mcl{P}_0$ be the original primal lattice and $\mcl{D}_0$ be the original dual lattice. Let $B_0$ be an MWPM on $\mcl{P}_0$ and $R_0$ be an MWPM on $\mcl{D}_0$. We reweight $\mcl{D}_0$ with $B_0$ and call it the first reweighted dual lattice $\mcl{D}_1$. Let $R_1$ be an MWPM on $\mcl{D}_1$. We reweight $\mcl{P}_0$ with $R_1$ and call it the first reweighted primal lattice $\mcl{P}_1$. We can use the similar way to construct $B_k$ and $R_k$, $k\in\mathbb{N}$.
Note that for $i>0$, $\mcl{D}_i$ is constructed by reweighting $\mcl{D}_0$ with $B_{i-1}$ and $\mcl{P}_i$ is constructed by reweighting $\mcl{P}_0$ with $R_i$.

When we have gotten $B_i$ and $R_i$ and try to calculate the total weight of them, we can not just calculate the weights of $B_i$ and $R_i$ both on the reweighted lattices and sum them up. One of them must be calculated on the original lattice and the other is calculated on the reweighted lattice reweighted with the first matching. Since $\mcl{P}_i$ is constructed based on $R_i$, to calculate the total weight of $B_i$ and $R_i$, we can calculate the weight of $R_i$ on $\mcl{D}_0$ first, and calculate that of $B_i$ on $\mcl{P}_i$, and then sum them up.

Let the weight of a matching $M$ on the $i$th reweighted lattice as $W_i(M)$. We can define the $i$th total weight as
\[T_i=W_i(B_i)+W_{0}(R_i),~i\geq 1.\]
For the case $i=0$, we need a different definition, since $W_{0}(B_0)$ is clearly not the weight of $B_0$ on the lattice reweighted with $R_0$. But since $\mcl{D}_1$ is the lattice reweighted with $B_0$, we can sum the weight of $B_0$ on $\mcl{P}_0$ and that of $R_0$ on $\mcl{D}_1$. Thus, the total weight of $B_0$ and $R_0$ is 
\[T_0 = W_0(B_0)+W_{1}(R_0).\] 
And example of the modified decoding process and examples of the above definitions can be seen in Fig. \ref{Whole decoding process}.  

\begin{figure}[h]
\centering
\begin{tikzpicture}[scale=0.5]
\node (syndrome) at (0,0.5)
{
	\begin{tikzpicture}[scale=0.6]
		\draw[step=1,dash pattern=on 2pt off 3pt,xshift=0.5cm, yshift=0.5cm](0,-0.5) 				grid(3,2.5);
		\draw[step=1,xshift=1cm](-0.5,0) grid(2.5,3);
 (1,1)rectangle(2,2);
		
		\filldraw[fill=blue, draw= blue] (1,2)circle(0.15);
		\filldraw[fill=blue, draw= blue] (2,0)circle(0.15);
		\filldraw[fill=blue, draw= blue] (2,1)circle(0.15);
		\filldraw[fill=blue, draw= blue] (3,0)circle(0.15);
		\filldraw[fill=blue, draw= blue] (3,1)circle(0.15);
		\filldraw[fill=blue, draw= blue] (3,2)circle(0.15);
		\filldraw[fill=red, draw= red]   (0.5,1.5)circle(0.15);
		\filldraw[fill=red, draw= red]   (0.5,2.5)circle(0.15);
		\filldraw[fill=red, draw= red]   (2.5,0.5)circle(0.15);
		\filldraw[fill=red, draw= red]   (2.5,2.5)circle(0.15);
		\filldraw[fill=red, draw= red]   (3.5,1.5)circle(0.15);
		\filldraw[fill=red, draw= red]   (3.5,2.5)circle(0.15);

		\node at (2,3.5) {The syndrome};
	\end{tikzpicture}
};

\node (P0) at (-3,-5){
	\begin{tikzpicture}[scale=0.6]
		\draw[step=1,dash pattern=on 2pt off 3pt,xshift=0.5cm, yshift=0.5cm](0,-0.5) 				grid(3,2.5);
		\draw[step=1,xshift=1cm](-0.5,0) grid(2.5,3);
 (1,1)rectangle(2,2);
		
		\filldraw[fill=blue, draw= blue] (1,2)circle(0.15);
		\filldraw[fill=blue, draw= blue] (2,0)circle(0.15);
		\filldraw[fill=blue, draw= blue] (2,1)circle(0.15);
		\filldraw[fill=blue, draw= blue] (3,0)circle(0.15);
		\filldraw[fill=blue, draw= blue] (3,1)circle(0.15);
		\filldraw[fill=blue, draw= blue] (3,2)circle(0.15);
		\filldraw[fill=red, draw= red]   (0.5,1.5)circle(0.15);
		\filldraw[fill=red, draw= red]   (0.5,2.5)circle(0.15);
		\filldraw[fill=red, draw= red]   (2.5,0.5)circle(0.15);
		\filldraw[fill=red, draw= red]   (2.5,2.5)circle(0.15);
		\filldraw[fill=red, draw= red]   (3.5,1.5)circle(0.15);
		\filldraw[fill=red, draw= red]   (3.5,2.5)circle(0.15);

		\draw [blue, line width = 1mm] (1,2) -- (3,2);
		\draw [blue, line width = 1mm] (2,0) -- (2,1);
		\draw [blue, line width = 1mm] (3,0) -- (3,1);
		
		\node at (2,3.4) {$(\mcl{P}_0,B_0)$};
	\end{tikzpicture}
};

\node (D0) at (3,-5){
	\begin{tikzpicture}[scale=0.6]
		\draw[step=1,dash pattern=on 2pt off 3pt,xshift=0.5cm, yshift=0.5cm](0,-0.5) 				grid(3,2.5);
		\draw[step=1,xshift=1cm](-0.5,0) grid(2.5,3);
 (1,1)rectangle(2,2);
		
		\filldraw[fill=blue, draw= blue] (1,2)circle(0.15);
		\filldraw[fill=blue, draw= blue] (2,0)circle(0.15);
		\filldraw[fill=blue, draw= blue] (2,1)circle(0.15);
		\filldraw[fill=blue, draw= blue] (3,0)circle(0.15);
		\filldraw[fill=blue, draw= blue] (3,1)circle(0.15);
		\filldraw[fill=blue, draw= blue] (3,2)circle(0.15);
		\filldraw[fill=red, draw= red]   (0.5,1.5)circle(0.15);
		\filldraw[fill=red, draw= red]   (0.5,2.5)circle(0.15);
		\filldraw[fill=red, draw= red]   (2.5,0.5)circle(0.15);
		\filldraw[fill=red, draw= red]   (2.5,2.5)circle(0.15);
		\filldraw[fill=red, draw= red]   (3.5,1.5)circle(0.15);
		\filldraw[fill=red, draw= red]   (3.5,2.5)circle(0.15);

		\draw [red, line width = 1mm] (0.5,1.5) -- (0.5,2.5);
		\draw [red, line width = 1mm] (2.5,2.5) -- (3.5,2.5);
		\draw [red, line width = 1mm] (2.5,0.5) -- (2.5,1.5);
		\draw [red, line width = 1mm] (2.5,1.5) -- (3.5,1.5);
		
		\node at (2,3.4) {$(\mcl{D}_0,R_0)$};
	\end{tikzpicture}
};

\node (T0) at (4.8,-8.8){
{$\begin{aligned}
T_0=&W_0(B_0)+W_1(R_0)\\
=&4+4=8
\end{aligned}$}
};

\node (D1) at (3,-12){
	\begin{tikzpicture}[scale=0.6]
		\draw[step=1,dash pattern=on 2pt off 3pt,xshift=0.5cm, yshift=0.5cm](0,-0.5) 				grid(3,2.5);
		\draw[step=1,xshift=1cm](-0.5,0) grid(2.5,3);
 (1,1)rectangle(2,2);
		
		\filldraw[fill=blue!30, draw= blue!30] (1,2)circle(0.15);
		\filldraw[fill=blue!30, draw= blue!30] (2,0)circle(0.15);
		\filldraw[fill=blue!30, draw= blue!30] (2,1)circle(0.15);
		\filldraw[fill=blue!30, draw= blue!30] (3,0)circle(0.15);
		\filldraw[fill=blue!30, draw= blue!30] (3,1)circle(0.15);
		\filldraw[fill=blue!30, draw= blue!30] (3,2)circle(0.15);
		\filldraw[fill=red, draw= red]   (0.5,1.5)circle(0.15);
		\filldraw[fill=red, draw= red]   (0.5,2.5)circle(0.15);
		\filldraw[fill=red, draw= red]   (2.5,0.5)circle(0.15);
		\filldraw[fill=red, draw= red]   (2.5,2.5)circle(0.15);
		\filldraw[fill=red, draw= red]   (3.5,1.5)circle(0.15);
		\filldraw[fill=red, draw= red]   (3.5,2.5)circle(0.15);

		\draw [blue!30, line width = 1mm] (1,2) -- (3,2);
		\draw [blue!30, line width = 1mm] (2,0) -- (2,1);
		\draw [blue!30, line width = 1mm] (3,0) -- (3,1);		
		
		\draw [red, line width = 1mm] (0.5,1.5) -- (0.5,2.5);
		\draw [red, line width = 1mm] (2.5,2.5) -- (3.5,2.5);
		\draw [red, line width = 1mm] (2.5,0.5) -- (3.5,0.5);
		\draw [red, line width = 1mm] (3.5,0.5) -- (3.5,1.5);
		
		\node at (2,3.4) {$(\mcl{D}_1,R_1)$};
		
	\end{tikzpicture}
};

\draw [-{Latex[length=3mm]}] (P0)--(D1);

\node (P1) at (-3,-12){
	\begin{tikzpicture}[scale=0.6]
		\draw[step=1,dash pattern=on 2pt off 3pt,xshift=0.5cm, yshift=0.5cm](0,-0.5) 				grid(3,2.5);
		\draw[step=1,xshift=1cm](-0.5,0) grid(2.5,3);
 (1,1)rectangle(2,2);
		
		\filldraw[fill=blue, draw= blue] (1,2)circle(0.15);
		\filldraw[fill=blue, draw= blue] (2,0)circle(0.15);
		\filldraw[fill=blue, draw= blue] (2,1)circle(0.15);
		\filldraw[fill=blue, draw= blue] (3,0)circle(0.15);
		\filldraw[fill=blue, draw= blue] (3,1)circle(0.15);
		\filldraw[fill=blue, draw= blue] (3,2)circle(0.15);
		\filldraw[fill=red!30, draw= red!30]   (0.5,1.5)circle(0.15);
		\filldraw[fill=red!30, draw= red!30]   (0.5,2.5)circle(0.15);
		\filldraw[fill=red!30, draw= red!30]   (2.5,0.5)circle(0.15);
		\filldraw[fill=red!30, draw= red!30]   (2.5,2.5)circle(0.15);
		\filldraw[fill=red!30, draw= red!30]   (3.5,1.5)circle(0.15);
		\filldraw[fill=red!30, draw= red!30]   (3.5,2.5)circle(0.15);		
		
		\draw [red!30, line width = 1mm] (0.5,1.5) -- (0.5,2.5);
		\draw [red!30, line width = 1mm] (2.5,2.5) -- (3.5,2.5);
		\draw [red!30, line width = 1mm] (2.5,0.5) -- (3.5,0.5);
		\draw [red!30, line width = 1mm] (3.5,0.5) -- (3.5,1.5);

		\draw [blue, line width = 1mm] (0.5,2) -- (1,2);
		\draw [blue, line width = 1mm] (3,2) -- (3.5,2);
		\draw [blue, line width = 1mm] (2,0) -- (2,1);
		\draw [blue, line width = 1mm] (3,0) -- (3,1);

		\node at (2,3.4) {$(\mcl{P}_1,B_1)$};
	\end{tikzpicture}
};

\node (T1) at (-4.8,-8.8){
{$\begin{aligned}
T_1=&W_1(B_1)+W_0(R_1)\\
=&2+4=6
\end{aligned}$}
};

\draw [-{Latex[length=3mm]}] (D1)--(P1);

\node (D2) at (3,-19){
	\begin{tikzpicture}[scale=0.6]
		\draw[step=1,dash pattern=on 2pt off 3pt,xshift=0.5cm, yshift=0.5cm](0,-0.5) 				grid(3,2.5);
		\draw[step=1,xshift=1cm](-0.5,0) grid(2.5,3);
 (1,1)rectangle(2,2);
		
		\filldraw[fill=blue!30, draw= blue!30] (1,2)circle(0.15);
		\filldraw[fill=blue!30, draw= blue!30] (2,0)circle(0.15);
		\filldraw[fill=blue!30, draw= blue!30] (2,1)circle(0.15);
		\filldraw[fill=blue!30, draw= blue!30] (3,0)circle(0.15);
		\filldraw[fill=blue!30, draw= blue!30] (3,1)circle(0.15);
		\filldraw[fill=blue!30, draw= blue!30] (3,2)circle(0.15);
		\filldraw[fill=red, draw= red]   (0.5,1.5)circle(0.15);
		\filldraw[fill=red, draw= red]   (0.5,2.5)circle(0.15);
		\filldraw[fill=red, draw= red]   (2.5,0.5)circle(0.15);
		\filldraw[fill=red, draw= red]   (2.5,2.5)circle(0.15);
		\filldraw[fill=red, draw= red]   (3.5,1.5)circle(0.15);
		\filldraw[fill=red, draw= red]   (3.5,2.5)circle(0.15);
		
		\draw [blue!30, line width = 1mm] (0.5,2) -- (1,2);
		\draw [blue!30, line width = 1mm] (3,2) -- (3.5,2);
		\draw [blue!30, line width = 1mm] (2,0) -- (2,1);
		\draw [blue!30, line width = 1mm] (3,0) -- (3,1);

		\draw [red, line width = 1mm] (0.5,1.5) -- (0.5,2.5);
		\draw [red, line width = 1mm] (2.5,0.5) -- (2.5,2.5);
		\draw [red, line width = 1mm] (3.5,1.5) -- (3.5,2.5);
		
		\node at (2,3.4) {$(\mcl{D}_2,R_2)$};
	\end{tikzpicture}
};

\draw [-{Latex[length=3mm]}] (P1)--(D2);

\node (P2) at (-3,-19){
	\begin{tikzpicture}[scale=0.6]
		\draw[step=1,dash pattern=on 2pt off 3pt,xshift=0.5cm, yshift=0.5cm](0,-0.5) 				grid(3,2.5);
		\draw[step=1,xshift=1cm](-0.5,0) grid(2.5,3);
 (1,1)rectangle(2,2);
		
		\filldraw[fill=blue, draw= blue] (1,2)circle(0.15);
		\filldraw[fill=blue, draw= blue] (2,0)circle(0.15);
		\filldraw[fill=blue, draw= blue] (2,1)circle(0.15);
		\filldraw[fill=blue, draw= blue] (3,0)circle(0.15);
		\filldraw[fill=blue, draw= blue] (3,1)circle(0.15);
		\filldraw[fill=blue, draw= blue] (3,2)circle(0.15);
		\filldraw[fill=red!30, draw= red!30]   (0.5,1.5)circle(0.15);
		\filldraw[fill=red!30, draw= red!30]   (0.5,2.5)circle(0.15);
		\filldraw[fill=red!30, draw= red!30]   (2.5,0.5)circle(0.15);
		\filldraw[fill=red!30, draw= red!30]   (2.5,2.5)circle(0.15);
		\filldraw[fill=red!30, draw= red!30]   (3.5,1.5)circle(0.15);
		\filldraw[fill=red!30, draw= red!30]   (3.5,2.5)circle(0.15);
		
		\draw [red!30, line width = 1mm] (0.5,1.5) -- (0.5,2.5);
		\draw [red!30, line width = 1mm] (2.5,0.5) -- (2.5,2.5);
		\draw [red!30, line width = 1mm] (3.5,1.5) -- (3.5,2.5);	

		\draw [blue, line width = 1mm] (0.5,2) -- (1,2);
		\draw [blue, line width = 1mm] (3,2) -- (3.5,2);
		\draw [blue, line width = 1mm] (2,0) -- (3,0);
		\draw [blue, line width = 1mm] (2,1) -- (3,1);
		
		\node at (2,3.4) {$(\mcl{P}_2,B_2)$};
	\end{tikzpicture}
};
\draw [-{Latex[length=3mm]}] (D2)--(P2);

\node (T2) at (-4.8,-15.7){
{$\begin{aligned}
T_2=&W_2(B_2)+W_0(R_2)\\
=&1+4=5
\end{aligned}$}
};

\end{tikzpicture}
\caption{An example of the modified decoding process.}
\label{Whole decoding process}
\end{figure}
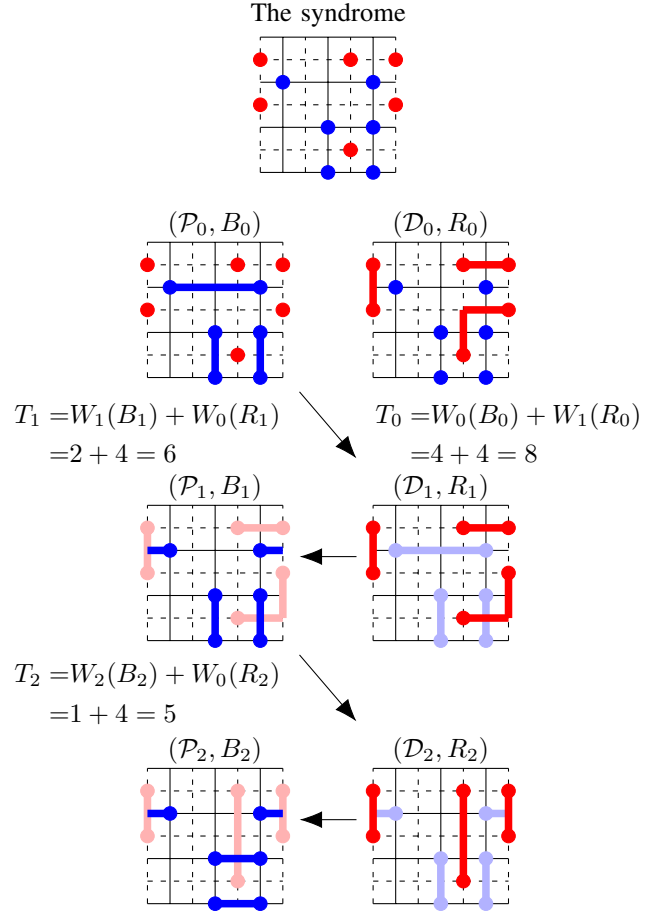

\begin{theorem}
$T_{i}\leq T_{i-1}$ for all $i\in\mathbb{N}$.
\end{theorem}

\begin{proof}
For an MWPM $M_P$ on the primal lattice and an MWPM $M_D$ on the dual lattice, there are two ways to calculate the total weight. The first one is summing the weight of $M_P$ on the original primal lattice and the weight of $M_D$ on the dual lattice reweighted with $M_P$. The second one is the reverse, i.e., summing the weight of $M_D$ on the original dual lattice and the weight of $M_P$ on the primal lattice reweighted with $M_D$. Therefore, for $i\geq1$, we have the following properties
\begin{equation} \label{eq:1}
W_0(B_{i-1})+W_i(R_i)=W_i(B_{i-1})+W_0(R_i)
\end{equation}
\begin{equation} \label{eq:2}
W_i(B_i)+W_0(R_i)=W_0(B_i)+W_{i+1}(R_i).
\end{equation}

Since the definition of $T_i$ are different for $i=0$ and $i\geq 1$, we need to discuss two cases. For $i=1$, since $R_1$ is an MWPM on $\mcl{D}_1$, we have $W_1(R_1)\leq W_1(R_0)$, then
\[W_0(B_0)+W_1(R_1)\leq W_0(B_0)+W_1(R_0)=T_0.\]
Since $W_0(B_0)+W_1(R_1)=W_1(B_0)+W_0(R_1)$, we have
\[W_1(B_0)+W_0(R_1)\leq W_0(B_0)+W_1(R_0)=T_0.\]
Since $B_1$ is an MWPM on $\mcl{P}_1$, we have $W_1(B_1)\leq W_1(B_0)$. Therefore,
\[T_1=W_1(B_1)+W_0(R_1)\leq W_0(B_0)+W_1(R_0)=T_0.\]
For $i\geq2$, let us start from $T_{i-1}=W_{i-1}(B_{i-1})+W_0(R_{i-1})$. We have $T_{i-1}=W_0(B_{i-1})+W_i(R_{i-1})$ by Equation \ref{eq:2}. Since $R_i$ is an MWPM on $\mcl{D}_i$, we have
\[W_0(B_{i-1})+W_i(R_i) \leq W_0(B_{i-1})+W_i(R_{i-1}) = T_{i-1}.\]
By Equation \ref{eq:1}, we have
\[W_0(B_{i-1})+W_i(R_i) = W_i(B_{i-1})+W_0(R_i).\]
Similarly, since $B_i$ is an MWPM on $\mcl{P}_i$, we have $W_i(B_i)\leq W_i(B_{i-1})$. Then,
\[T_{i}=W_i(B_i)+W_0(R_i)\leq W_i(B_{i-1})+W_0(R_i) \leq T_{i-1}.\]
\end{proof}

Here we need to indicate that this modification does not guarantee the minimum total weight result. Take Fig. \ref{rewight example} as an example. If the MWPM we find on the primal lattice is as Fig. \ref{counter example of the algorithm}, then this method will fail to give the correction pattern with minimum total weight.

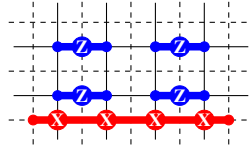
\begin{figure}[h]
     \centering
		\begin{tikzpicture}[scale=0.65, every node/.style={scale=0.7}]
		\draw[step=1](0.1,0.1) grid(4.9,2.9);
		%\draw[dotted] (0,0) -- (5,0);
		%\draw[dotted] (5,0) -- (5,3);
		\draw[step=1,dash pattern=on 2pt off 2pt,xshift=0.5cm, yshift=0.5cm](-0.5,-0.5) grid(4.5, 2.5);

		\filldraw [fill=red,draw = red](0.5,0.5) circle(0.1);
		\filldraw [fill=red,draw = red](4.5,0.5) circle(0.1);
		\foreach \x in {1,2,3,4}
			\foreach \y in {1,2}
				\filldraw [fill=blue,draw = blue](\x,\y) circle(0.1);
		\draw [red,line width=1mm] (0.5,0.5) -- (4.5,0.5);
		\draw [blue,line width=1mm] (1,1) -- (2,1);
		\filldraw [fill=blue,draw = blue](1.5,1) circle(0.2);
			\node [white] at (1.5,1) {\textbf{Z}};
		\draw [blue,line width=1mm] (1,2) -- (2,2);
		\filldraw [fill=blue,draw = blue](1.5,2) circle(0.2);
			\node [white] at (1.5,2) {\textbf{Z}};
		\draw [blue,line width=1mm] (3,1) -- (4,1);
		\filldraw [fill=blue,draw = blue](3.5,1) circle(0.2);
			\node [white] at (3.5,1) {\textbf{Z}};
		\draw [blue,line width=1mm] (3,2) -- (4,2);
		\filldraw [fill=blue,draw = blue](3.5,2) circle(0.2);
			\node [white] at (3.5,2) {\textbf{Z}};
		
		\foreach \x in {1,2,3,4}{
			\filldraw [fill=red,draw = red](\x,0.5) circle(0.2);
			\node [white] at (\x,0.5) {\textbf{X}};
		}
	\end{tikzpicture}
    \caption{For the syndrome in Fig. \ref{rewight example}, if $Z$-type errors are decoded as this, then this algorithm will not give us the correction pattern with minimum total weight. It's an example shows that the IRMWPM decoder doesn't guarantee the results of maximum likelihood decoding over depolarizing channel. }
    \label{counter example of the algorithm}
\end{figure}

In Section II, we do not discuss the time complexity of constructing the syndrome node graph, since the shortest path of any two syndrome nodes can be obtained in $O(1)$. However, the shortest path between two syndrome nodes on a reweighted lattice is not that clear. To find the shortest paths between all pairs of nodes in a graph, we can use Floyd-Warshall algorithm or use Dijkstra's algorithm on each pair of nodes. For a graph with $n$ nodes, the time complexity is $O(n^3)$ for both methods. Therefore, the time complexity of constructing syndrome node graphs is $O(n^3)$.

We will see in Section IV that it is rare to need more than $5$ iterations for lattices smaller than $30\times30$. Therefore, we can neglect how many iterations are used in the calculation of the total time complexity and the time complexity of the IRMWPM decoder is still $O(n^3)$.

\section{Simulation Results}

Since we will repeat the same process more that one time, we need to set a stopping criterion. Stopping the iterations as soon as the error weight stops decreasing is not good enough because it is possible that the error weight stays at a particular value for a few iterations and then drops again. Suppose that the subroutine we use to find MWPMs will give us the same results for two same complete graphs. We can use whether there is a previous correction pattern is the same as the current one as the stopping criterion. 

Here, we show the decoding performance of three different cases in Fig. \ref{three cases}. In the first one, we only apply MWPM decoding without any reweighting. In the second one, we reweight the dual lattice only one time, i.e., using $B_0$ and $R_1$ as the correction. In the third one, iterations will continue until the newest MWPM is the same as one of the previous MWPMs. We can see that the logical error rate decreases as more iterations are applied. 

%%%%%%%%%%%%%%%%%%%%%%%%%%%%%%%%%%%%%%%%%%
%%%%%%%%%%%%%%%%%%%%%%%%%%%%%%%%%%%%%%%%%%

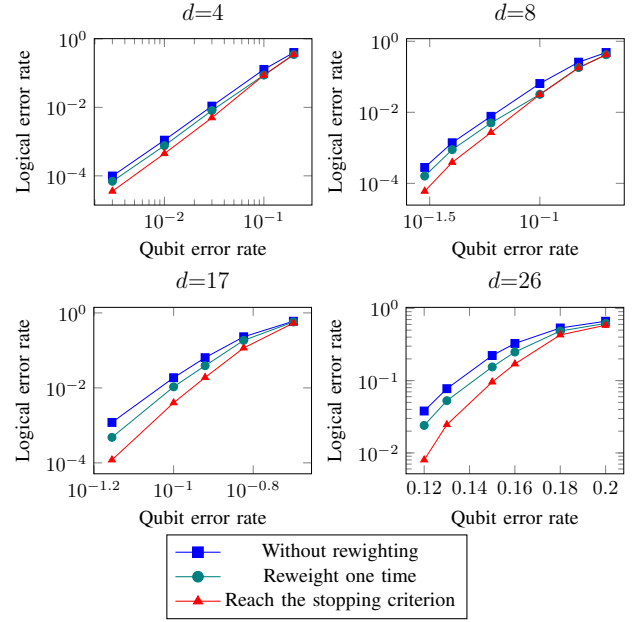
\begin{figure}[h]
     \centering
		\begin{tikzpicture}[scale= 0.75]
		\begin{axis}[title = {\large $d{=}4$},
			name=d_2,
			ymode=log,
        	xlabel=Qubit error rate,
        	ylabel=Logical error rate,
        	xmode=log,
        	width=0.3\textwidth,
            height=0.25\textwidth,
            legend style={nodes={scale=0.75, transform shape}},
        	legend entries={{Without rewighting},{Reweight one time},{Reach the stopping criterion}},
        	legend to name=named
    	]

    	\addplot [mark=square*,color=blue] plot coordinates {
        	(0.2,  39965/100000)
        	%(0.15, 26241/100000)
        	(0.10, 12820/100000)
        	(0.03, 5432/500000)
        	(0.01, 550/500000)
        	(0.003,988/10000000)
        	%(0.001,27/3000000)
    	}; 

    	\addplot [mark=*, color=teal] plot coordinates {
        	(0.2,  34893/100000)
        	%(0.15, 20719/100000)
        	(0.10, 8747/100000)
        	(0.03, 4038/500000)
        	(0.01, 381/500000)
        	(0.003,692/10000000)
        	%(0.001,20/3000000)
    	}; 
    	\addplot [mark=triangle*, color=red] plot coordinates {
        	(0.2,  34732/100000)
        	%(0.15, 20562/100000)
        	(0.10, 8667/100000)
        	(0.03, 2480/500000)
        	(0.01, 223/500000)
        	(0.003,356/10000000)
        	%(0.001,12/3000000)
    	};
    	\end{axis}
		\end{tikzpicture}		
		\begin{tikzpicture}[scale= 0.75]
		\begin{axis}[title = {\large $d{=}8$},
			ymode=log,
        	xlabel=Qubit error rate,
        	ylabel=Logical error rate,
        	xmode=log,
        	width=0.3\textwidth,
          height=0.25\textwidth
    	]

    	\addplot [mark=square*,color=blue] plot coordinates {
        	(0.2,  4815/10000)
        	(0.15, 2577/10000)
        	(0.10, 646/10000)
        	(0.06, 385/50000)
        	(0.04, 25/18000)
        	(0.03, 14/50000)
    	}; 

    	\addplot [mark=*, color=teal] plot coordinates {
        	(0.2,  4174/10000)
        	(0.15, 1815/10000)
        	(0.10, 320/10000)
        	(0.06, 251/50000)
        	(0.04, 16/18000)
        	(0.03, 8/50000)

    	}; 
    	\addplot [mark=triangle*, color=red] plot coordinates {
        	(0.2,  4159/10000)
        	(0.15, 1788/10000)
        	(0.10, 315/10000)
        	(0.06, 135/50000)
        	(0.04, 7/18000)
        	(0.03, 3/50000)
    	};
    	\end{axis}
		\end{tikzpicture}

		\begin{tikzpicture}[scale= 0.75]
		\begin{axis}[title = {\large $d{=}17$},
			ymode=log,
        	xlabel=Qubit error rate,
        	ylabel=Logical error rate,
        	xmode=log,
        	width=0.3\textwidth,
            height=0.25\textwidth
    	]

    	\addplot [mark=square*,color=blue] plot coordinates {
        	(0.2,  1199/2000)
        	(0.15, 231/1000)
        	(0.12, 64/1000)
        	(0.10, 28/1500)
        	(0.07, 10/8345)
    	}; 

    	\addplot [mark=*, color=teal] plot coordinates {
        	(0.2,  1146/2000)
        	(0.15, 185/1000)
        	(0.12, 39/1000)
        	(0.10, 16/1500)
        	(0.07, 4/8345)
    	}; 
    	\addplot [mark=triangle*,color=red] plot coordinates {
        	(0.2,  1070/2000)
        	(0.15, 116/1000)
        	(0.12, 19/1000)
        	(0.10, 6/1500)
        	(0.07, 1/8345)
    	};
    	\end{axis}
		\end{tikzpicture}
		\begin{tikzpicture}[scale= 0.75]
		\begin{axis}[title = {\large $d{=}26$},
			ymode=log,
        	xlabel=Qubit error rate,
        	ylabel=Logical error rate,
        	width=0.3\textwidth,
            height=0.25\textwidth,
        	]

    	\addplot [mark=square*,color=blue] plot coordinates {
        	(0.2,  1261/1900)
        	(0.18, 455/850)
        	(0.16, 684/2090)
        	(0.15, 891/4000)
        	(0.13, 151/1950)
			(0.12, 19/500 )
    	}; 

    	\addplot [mark=*, color=teal] plot coordinates {
        	(0.2,  1176/1900)
        	(0.18, 413/850)
        	(0.16, 520/2090)
        	(0.15, 619/4000)
        	(0.13, 103/1950)
        	(0.12, 12/500 )
    	}; 
    	\addplot [mark=triangle*, color=red] plot coordinates {
        	(0.2,  1110/1900)
        	(0.18, 365/850)
        	(0.16, 356/2090)
        	(0.15, 383/4000)
        	(0.13, 48/1950)
        	(0.12, 4/500 )
    	};
    	\end{axis}
		\end{tikzpicture}
		\\
\ref{named}
    \caption{The decoding simulation of the surface codes with mixed boundaries.}
    \label{three cases}
\end{figure}

Now we want to know how many iterations do we need. Empirically, it is rarely over $5$ when the lattice is smaller than $30\times30$. The counting of the extra iterations starts from using $B_1$ to reweight the dual lattice. Using $B_0$ to reweight the dual lattice and using $R_1$ to reweight the primal lattice are not viewed as extra iterations since the stopping criterion cannot be met without $B_1$. Fig. \ref{extra iterations} shows the distribution of how many extra iterations do the surface codes with mixed boundaries need for different code distances.

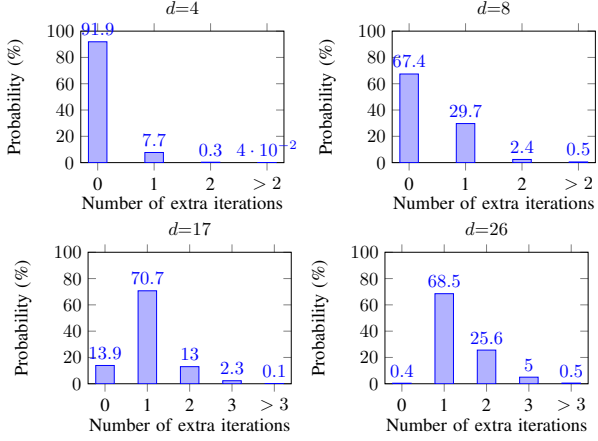
\begin{figure}[ht]
     \centering
		\begin{tikzpicture}[scale=0.7]
		\begin{axis} [ title = $d{=}4$,
			ybar,
			nodes near coords ,
			symbolic x coords={0,1,2,$>2$},
			xtick=data,
			ylabel={Probability (\%)},
			xlabel={Number of extra iterations},
			width=0.3\textwidth,
            height=0.225\textwidth,
            ymin=0,ymax=100,
		]
		\addplot coordinates {
		   (0,91.9) 
 		   (1,7.7) 
 		   (2,0.3) 
 		   ($>2$,0.04)
		};
		\end{axis}
		\end{tikzpicture}
		\begin{tikzpicture}[scale=0.7]
		\begin{axis} [ title = $d{=}8$,
			ybar,
			width=0.5\textwidth,
            height=0.5\textwidth,
			nodes near coords ,
			symbolic x coords={0,1,2,$>2$},
			xtick=data,
			ylabel={Probability (\%)},
			xlabel={Number of extra iterations},
			width=0.3\textwidth,
            height=0.225\textwidth,
            ymin=0,ymax=100,
		]
		\addplot coordinates {
		   (0,67.4) 
 		   (1,29.7) 
 		   (2,2.4) 
 		   ($>2$,0.5)
		};
		\end{axis}
		\end{tikzpicture}
		\\
		\begin{tikzpicture}[scale= 0.7]
		\begin{axis} [title = $d{=}17$,
			ybar,
			nodes near coords,
			symbolic x coords={0,1,2,3,$>3$},
			xtick=data,
			ylabel={Probability (\%)},
			xlabel={Number of extra iterations},
			width=0.3\textwidth,
            height=0.225\textwidth,
            ymin=0,ymax=100,
		]
		\addplot coordinates {
		   (0,13.9) 
 		   (1,70.7) 
 		   (2,13.0) 
 		   (3,2.3)
 		   ($>3$,0.1)
		};
		\end{axis}
		\end{tikzpicture}
		\begin{tikzpicture}[scale= 0.7]
		\begin{axis} [title = $d{=}26$,
			ybar,
			nodes near coords ,
			symbolic x coords={0,1,2,3,$>3$},
			xtick=data,
			ylabel={Probability (\%)},
			xlabel={Number of extra iterations},
			width=0.3\textwidth,
            height=0.225\textwidth,
            ymin=0,ymax=100,
		]
		\addplot coordinates {
		   (0,0.4) 
 		   (1,68.5)
 		   (2,25.6) 
 		   (3,5.0)
 		   ($>3$,0.5)
		};
		\end{axis}
		\end{tikzpicture}
    \caption{The distribution of how many extra iterations do surface codes with mixed boundaries need when the qubit error rate is $0.1$. The average of these four cases are $0.08$, $0.34$, $1.04$, and  $1.32$, respectively.}
    \label{extra iterations}
\end{figure}

For a surface code, we may want to increase the size of the lattice to lower the logical error rate, but the larger the lattice is, the more errors may be introduced into the system and the logical error rate increases. Threshold is an index of a surface code decoder. The logical error rate increases as the size of the code gets larger and larger when the qubit error rate is greater than the threshold. For the surface codes with mixed boundaries, the thresholds of the MWPM decoder is $15.5\%$ and that of the IRMWPM decoder is improved to $17.0\%$, as shown in Fig. \ref{Threshold Mixed boundaries}. Similar effects are observed for a surface code with a hole, and the thresholds of the MWPM decoder and IRMWPM decoder are  $14.2\%$ and  $15.3\%$, as shown in Fig. \ref{Threshold}. The thresholds and complexity of various decoders are provided in Table \ref{thresholds table}.

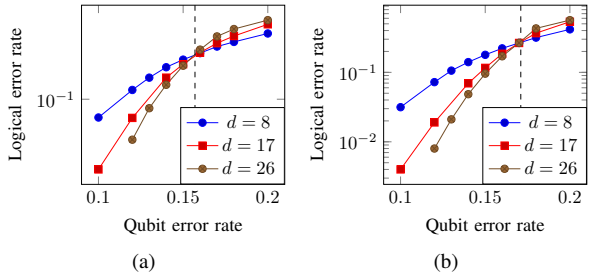
\begin{figure}[h]
     \centering
     \subfloat[]{
		\begin{tikzpicture}[scale= 0.7]
		\begin{axis}[
			ymode=log,
        	xlabel=Qubit error rate,
        	ylabel=Logical error rate,
        	width=0.3\textwidth,
            height=0.275\textwidth,
        	legend style={at={(1,0)},anchor=south east}
    	]
    	\addplot plot coordinates {
        	(0.2,     4815/10000)
        	(0.18,    3936/10000)
        	(0.17,    3523/10000)
        	(0.16,    3010/10000)
        	(0.15,    2577/10000)
        	(0.14,    2148/10000)
        	(0.13,    1675/10000)
        	(0.12,    1248/10000)
        	(0.10,    646/10000)
    	};
    	%d=16
    	\addplot plot coordinates {
        	(0.2,     1199/2000)
        	(0.18,    453/1000)
        	(0.17,    382/1000)
        	(0.16,    305/1000)
        	(0.15,    231/1000)
        	(0.14,    335/2000)
        	(0.12,    64/1000)
        	(0.10,    28/1500)
    	};
    	d=25
    	\addplot plot coordinates {
    		(0.20,    265/400)
        	(0.18,    455/850)
        	(0.17,    1149/2560)
        	(0.16,    684/2090)
        	(0.15,    891/4000)
        	(0.14,    233/1650)
        	(0.13,    234/2900 )	
        	(0.12,    19/500)
        	%(0.10,    5/500)
    	};
    	\draw [dashed] ({axis cs:0.157,0}|-{rel axis cs:0,1}) -- ({axis cs:0.157,0}|-{rel axis cs:0,0});
    	%\legend{$d=3$\\$d=7$\\$d=16$\\}
    	\legend{$d=8$\\$d=17$\\$d=26$\\}
    	\end{axis}
		\end{tikzpicture}
	\label{Threshold Mixed boundaries MWPM}		
		}
    \subfloat[]{
		\begin{tikzpicture}[scale= 0.7]
		\begin{axis}[
			ymode=log,
        	xlabel=Qubit error rate,
        	ylabel=Logical error rate,
        	width=0.3\textwidth,
            height=0.275\textwidth,
        	legend style={at={(1,0)},anchor=south east}
    	]
     % d=7
    	\addplot plot coordinates {
        	(0.2,     4159/10000)
        	(0.18,    3164/10000)
        	(0.17,    2651/10000)
        	(0.16,    2215/10000)
        	(0.15,    1788/10000)
        	(0.14,    1404/10000)
        	(0.13,    1059/10000)
        	(0.12,    723/10000)
        	(0.10,    315/10000)
    	}; 
    	% d=16
    	\addplot plot coordinates {
        	(0.2,     1070/2000)
        	(0.18,    367/1000)
        	(0.17,    265/1000 )
        	(0.16,    185/1000)
        	(0.15,    116/1000)
        	(0.14,    139/2000)
        	(0.12,    19/1000)
        	(0.10,    6/1500)
    	};
    	d=25
    	\addplot plot coordinates {
        	(0.2,     226/400)
        	(0.18,    365/850)
        	(0.17,    681/2560)
        	(0.16,    356/2090)
        	(0.15,    383/4000)
        	(0.14,    80/1650 )
        	(0.13,    20/950 )
        	(0.12,    4/500 )
    	};
		\draw [dashed] ({axis cs:0.171,0}|-{rel axis cs:0,1}) -- ({axis cs:0.171,0}|-{rel axis cs:0,0});
    	%\legend{$d=3$\\$d=7$\\$d=16$\\}
    	\legend{$d=8$\\$d=17$\\$d=26$\\}
    	\end{axis}
		\end{tikzpicture}
	\label{Threshold Mixed boundaries IRMWPM}	
	}
    \caption{Decoding performances of (a) MWPM decoder and (b) IRMWPM decoder on the surface codes with mixed boundaries.}
    \label{Threshold Mixed boundaries}
\end{figure}

%%%%%%%%%%%%%%%%%%%%%%%%%%%%%%%%%%%%%%%%%%%%%%%%%%%%%%
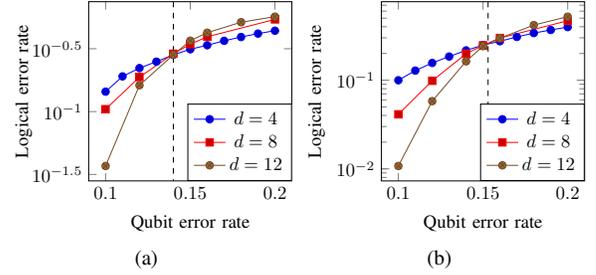
\begin{figure}[h]
     \centering
     \subfloat[]{
		\begin{tikzpicture}[scale= 0.7]
		\begin{axis}[
			ymode=log,
        	xlabel=Qubit error rate,
        	ylabel=Logical error rate,
        	width=0.3\textwidth,
            height=0.275\textwidth,
        	legend style={at={(1,0)},anchor=south east}
    	]
    	%d=2
      	%\addplot plot coordinates {
        %	(0.2,     0.44639)
        %	(0.19,    0.42716)
        %	(0.18,    0.41347)
        %	(0.17,    0.40067)
        %	(0.16,    0.37998)
        %	(0.15,    0.3643)
        %	(0.12,   14973/50000)
        %	(0.10,   2570/10000)
    	%};
    	%d=4
    	\addplot  plot coordinates {
        	(0.2,     0.44003)
        	(0.19,    0.41733)
        	(0.18,    0.39208)
        	(0.17,    0.36583)
        	(0.16,    0.33725)
        	(0.15,    0.31318)
        	(0.14,    0.28295)
        	(0.13,    0.24969)
        	(0.12,    0.22178)
        	(0.11,    0.19063)
        	(0.10,   39026/270844)
    	};
    	%d=6
    	%\addplot plot coordinates {
        %	(0.2,    0.512)
        %	(0.19,   0.4861)
        %	(0.18,   0.4568)
        %	(0.17,   0.4258)
        %	(0.16,   0.391)
        %	(0.15,   0.3664)
        %	(0.14,   0.3204)
        %	(0.13,   0.2718)
        %	(0.12,   0.2336)
        %	(0.10,   487/3924)
    	%};
    	d=8
    	\addplot plot coordinates {
        	(0.2,    1082/2000)
        	(0.16,   3942/10000)
        	(0.15,   1718/5000)
        	(0.14,   2889/10000)
        	(0.12,   940/5000)
        	(0.10,   832/7960)
    	};
    	%d=12
    	\addplot plot coordinates {
        	(0.2,    115/202)
        	(0.18,   257/500)
        	(0.16,   844/1989)
        	(0.15,   740/2019)
        	(0.14,   283/1000)
        	(0.12,   81/500)
        	(0.10,   24/650)
    	};
    	\draw [dashed] ({axis cs:0.140,0}|-{rel axis cs:0,1}) -- ({axis cs:0.140,0}|-{rel axis cs:0,0});
    	\legend{$d=4$\\$d=8$\\$d=12$\\}
    	\end{axis}
		\end{tikzpicture}}
    \subfloat[]{
		\begin{tikzpicture}[scale= 0.7]
		\begin{axis}[
			ymode=log,
        	xlabel=Qubit error rate,
        	ylabel=Logical error rate,
        	width=0.3\textwidth,
            height=0.275\textwidth,
        	legend style={at={(1,0)},anchor=south east}
    	]
    	%d=2
      	%\addplot plot coordinates {
        %	(0.2,     0.40376)
        %	(0.19,    0.38099)
        %	(0.18,    0.36584)
        %	(0.17,    0.34895)
        %	(0.16,    0.32592)
        %	(0.15,    0.3087)
        %	(0.12,   11996/50000)
        %	(0.10,   1936/10000)
    	%};
    	%d=4
    	\addplot plot coordinates {
        	(0.2,     0.39674)
        	(0.19,    0.36777)
        	(0.18,    0.33966)
        	(0.17,    0.30716)
        	(0.16,    0.27476)
        	(0.15,    0.2482)
        	(0.14,    0.21661)
        	(0.13,    0.18474)
        	(0.12,    0.15638)
        	(0.11,    0.1282)
        	(0.10,   27031/270844)
    	}; 
    	% d=6
    	%\addplot plot coordinates {
        %	(0.2,    0.436)
        %	(0.19,   0.4291)
        %	(0.18,   0.407)
        %	(0.17,   0.3604)
        %	(0.16,   0.3148)
        %	(0.15,   0.2772)
        %	(0.14,   0.2265)
        %	(0.13,   0.187)
        %	(0.12,   0.1464)
        %	(0.10,   298/3924)
    	%}; 
    	% d=8
    	\addplot plot coordinates {
        	(0.2,    935/2000)
        	(0.16,   2988/10000)
        	(0.15,   1235/5000)
        	(0.14,   1985/10000)
        	(0.12,   492/5000)
        	(0.10,   328/7960)
    	}; 
    	% d=12
    	\addplot plot coordinates {
        	(0.2,    105/202)
        	(0.18,   209/500)
        	(0.16,   591/1989)
        	(0.15,   488/2019)
        	(0.14,   163/1000)
        	(0.12,   29/500)
        	(0.10,   7/650)
    	};
		\draw [dashed] ({axis cs:0.153,0}|-{rel axis cs:0,1}) -- ({axis cs:0.153,0}|-{rel axis cs:0,0});
    	\legend{$d=4$\\$d=8$\\$d=12$\\}
    	\end{axis}
    	
		\end{tikzpicture}}
    \caption{Decoding performances of (a) MWPM decoder and (b) IRMWPM decoder on the surface codes with a hole.}
    \label{Threshold}
\end{figure}

%%%%%%%%%%%%%%%%%%%%%%%%%%%%%%%%%%%%%%%%%%%%%%%%%%%%%%

\begin{table}[h]
\renewcommand{\arraystretch}{1.3}
\caption{The thresholds of various decoders on surface codes over depolarizing errors}
\label{thresholds table}
\centering
\begin{tabular}{|c|c|c|}
\hline
Decoder & Threshold & Complexity\\
\hline
UF \cite{UF1} & -- & $O(n)$\\
\hline
MBP \cite{BP1} & $14.5\%–16\%$ & $O(n \log \log n)$\\
\hline
MWPM & $15.5\%$ & $O(n^3)$\\
\hline
IRMWPM & $17\%$ & $O(n^3)$\\
\hline
BP-MWPM \cite{BP2} & $17.84\%$ & $O(n^3)$\\
\hline
MPS \cite{MPS1} & $17\%–18.5\%$ & $O(n\chi^3)$\\
\hline
TN \cite{TN} & $18.81\%$ & $O(n\log n+n\chi^3)$\\
\hline
\end{tabular}
\end{table}
%%%%%%%%%%%%%%%%%%%%%%%%%%%%%%%%%%%%%%%%%%%%%%%%%%%%%%

\section{Conclusion}
We propose a modification to the conventional MWPM decoding of the surface codes to deal with the noise in depolarizing channels where the bit-flip errors and the phase-flip errors are correlated. 
Our method is mainly based on repeatedly using an MWPM on one lattice to reweight the other lattice to get a correction pattern with possibly less total weight. In this paper, we prove that the total weight will never increase when we repeat this reweighting process, and we present the simulation results of both the surface codes with mixed boundaries and the surface codes with a hole to show the improvement of the decoding performances.

\end{document}